\documentclass[a4paper,12pt]{article}
\usepackage{amssymb}
\usepackage{amsthm}
\usepackage{amsmath}
\usepackage[utf8]{inputenc}
\usepackage{array}
\usepackage[T1]{fontenc}
\usepackage{lmodern}
\usepackage{tikz}
\usepackage{vmargin}
\usepackage{enumitem}
\usepackage{authblk,tikz}

\setmarginsrb{2.5cm}{1cm}{2.5cm}{1cm}{1cm}{1cm}{1cm}{1cm}

\newtheorem{theorem}{Theorem}
\newtheorem*{theo*}{Theorem}
\newtheorem{corollary}[theorem]{Corollary}
\newtheorem{lemma}[theorem]{Lemma}

\newtheorem{obs}[theorem]{Observation}
\newtheorem{conjecture}[theorem]{Conjecture}
\newtheorem{claim}{Claim}[theorem]

\DeclareMathOperator{\unvd}{unvd}

\DeclareMathOperator{\hit}{hit}
\DeclareMathOperator{\spc}{spc}

\newcommand{\vstub}{n + 36k^2 - 140 k+ 124}
\newcommand{\val}{n+ 144k^2 - 280k + 124}
\newcommand{\ra}{\rightarrow}

\newenvironment{subproof}{\par\noindent {\it Subproof}.\ }{\hfill$\lozenge$\par\vspace{11pt}}

\title{On the unavoidability of oriented trees}

\author[1]{Fran\c{c}ois Dross}
\author[2]{Fr\'ed\'eric Havet}

\affil[1]{CNRS, Universit\'e C\^ote d'Azur, Inria, I3S, France}
\affil[2]{Universit\'e C\^ote d'Azur, CNRS, Inria, I3S, France}

\begin{document}

\maketitle

\begin{abstract}
A digraph is {\it $n$-unavoidable} if it is contained in every tournament of order $n$.
We first prove that every arborescence of order $n$ with $k$ leaves is $(n+k-1)$-unavoidable.
We then prove that every oriented tree of order $n$ ($n\geq 2$) with $k$ leaves is $(\frac{3}{2}n+\frac{3}{2}k -2)$-unavoidable and $(\frac{9}{2}n -\frac{5}{2}k -\frac{9}{2})$-unavoidable, and thus $(\frac{21}{8} n- \frac{47}{16})$-unavoidable.
Finally, we prove that every oriented tree of order $n$ with $k$ leaves is $(\val)$-unavoidable.
\end{abstract}

\section{Introduction}
%%%%%%%%%%%%

A {\bf tournament} is an orientation of a complete graph.
A digraph is {\bf $n$-unavoidable} if it is contained (as a subdigraph) in every tournament of order $n$. 
The {\bf unavoidability} of a digraph $D$, denoted by $\unvd(D)$, is the minimum integer $n$ such that $D$ is $n$-unavoidable. 
It is well-known that the transitive tournament of order $n$ is $2^{n-1}$-unavoidable and thus every acyclic digraph of order $n$ is $2^{n-1}$-unavoidable.
However, for acyclic digraphs with few arcs better bounds are expected. 
Special attention has been devoted to {\bf oriented paths} and {\bf oriented trees}, which are orientations of paths and trees respectively.

It started with R\'edei's Theorem~\cite{Rede34} which states that the unavoidabilty of $\vec{P}_n$, the directed path on $n$ vertices, is $n$:  $\unvd(\vec{P}_n)=n$.
In 1971, Gr\"unbaum studied the {\bf antidirected paths} that are oriented paths in which every vertex has either in-degree $0$ or out-degree $0$ (in other words, two consecutive edges are oriented in opposite ways). He proved~\cite{Gru71} that the unavoidability of an antidirected path of order $n$ is $n$ unless $n=3$ (in which case it is not contained in the directed $3$-cycle $\vec{C}_3$) or $n=5$ (in which case it is not contained in the regular tournament of order $5$) or $n=7$ (in which case it is not contained in the Paley tournament of order $7$).   
The same year,  Rosenfeld~\cite{Ros71} gave an easier  proof
and conjectured  that  there is a smallest integer $N_P>7$ such that $\unvd(P) = |P|$ for every oriented path of order at least $N_P$.   
The condition $N_P>7$ results from Gr\"unbaum's counterexamples. 
Several papers gave partial answers to this conjecture \cite{AlRo81,For73,STR80} until Rosenfeld's conjecture was verified by Thomason, who proved 
in~\cite{Tho86} that $N_P$ exists and is less than $2^{128}$. Finally, Havet and Thomass\'e~\cite{HaTh00}, showed that $\unvd(P)=|P|$ for every oriented path $P$ except the antidirected paths of order $3$, $5$, and $7$.

Regarding oriented trees, Sumner (see~\cite{ReWo83}) made the following celebrated conjecture.
\begin{conjecture}\label{conj:Sumner}
 Every oriented tree of order $n>1$ is $(2n-2)$-unavoidable.
 \end{conjecture}
The first linear bound was given by 
H\"aggkvist and Thomason~\cite{HaTh91}. Following improvements of Havet~\cite{Hav02} and Havet and Thomass\'e~\cite{HaTh00b}, El Sahili \cite{ElS04c} used the notion of median order, first used as a tool for Sumner's conjecture in \cite{HaTh00b}, and proved that every oriented tree of order $n$ ($n\geq 2$) is $(3n-3)$-unavoidable.
Recently, K\"uhn, Mycroft and Osthus~\cite{KMO11b} proved that
Sumner's conjecture is true for all sufficiently large $n$.
Their complicated proof makes use of the directed version of the Regularity Lemma and of results and ideas from a recent paper by the same authors~\cite{KMO11a}, in which an approximate version of the conjecture was proved.
In~\cite{HaTh00b}, Havet and Thomass\'e also proved that Sumner's conjecture holds for arborescences. An {\bf in-arborescence}, (resp.~{\bf out-arborescence}) is an oriented tree in which all arcs are oriented towards (resp.~away from) a fixed vertex called the root.
An {\bf arborescence} is either an in-arborescence or an out-arborescence.

If true, Sumner's conjecture would be tight. Indeed, the {\bf out-star} $S^+_n$, which is the digraph on $n$ vertices consisting of a vertex dominating the $n-1$ others, is not contained in the regular tournaments of order $2n-3$. However, such digraphs have many leaves.
Therefore Havet and Thomass\'e (see \cite{Hav03}) made the following stronger conjecture than Sumner's one.

\begin{conjecture}\label{conj:n+k-1}
 Every oriented tree of order $n$ with $k$ leaves is $(n+k-1)$-unavoidable.
 \end{conjecture}
 
As an evidence to this conjecture,  H\"aggkvist and Thomason~\cite{HaTh91} proved the existence of a minimal
function $g(k)\leq 2^{512k^3}$ such that every tree of order $n$ with $k$
leaves is $(n+g(k))$-unavoidable.
Trees with two leaves are paths, so the above-mentioned results imply that Conjecture~\ref{conj:n+k-1} is true when $k=2$ and
Ceroi and Havet~\cite{Ceha04} showed that it holds for $k=3$.
Havet~\cite{Hav03} also showed Conjecture~\ref{conj:n+k-1} for a large class of trees.

\subsection{Our results}
%%%%%%%%%%%%%

In Section~\ref{sec:arbo}, we prove Conjecture~\ref{conj:n+k-1} for arborescences.
\begin{theorem}\label{thm:arbo}
 Every arborescence of order $n$ with $k$ leaves is $(n+k-1)$-unavoidable.
 \end{theorem}

Using this result, in Section~\ref{sec:few}, we derive the following.
\begin{theorem} \label{t_tree2}
Every oriented tree of order $n$ with  $k$ leaves is $(\frac32 n+ \frac32 k-2)$-unavoidable.
\end{theorem}

This result gives us a good bound for trees with few leaves.
In particular, it implies Sumner's conjecture for trees in which at most one third of the vertices are leaves.

\begin{corollary}
Every oriented tree of order $n$ with at most $\frac n3$ leaves is $(2n-2)$-unavoidable.
\end{corollary}

Then, in Section~\ref{sec:many}, we give the following upper bound on the unvoidability of trees, which is good for trees with many leaves.

\begin{theorem} \label{t_leaves}
Every oriented tree with $n\ge 3$ vertices and $k$ leaves is $(\frac92n - \frac52k - \frac92)$-unavoidable.
\end{theorem}

Theorems~\ref{t_tree2} and  \ref{t_leaves} yield the best bound towards Sumner's conjecture:
\begin{corollary}
Every oriented tree of order $n \ge 2$ is $\left ( \frac{21}{8} n - \frac{47}{16}\right )$-unavoidable.
\end{corollary}
\begin{proof}
The value of $\min \left (\frac32n+\frac32k-2, \frac92n - \frac52k - \frac92 \right)$ is maximal when $\frac32n+\frac32k-2 = \frac92n - \frac52k - \frac92$, that is when $k = \frac{6n-5}{8}$. In this case $\frac32n+\frac32k-2 = \frac{21}{8} n- \frac{47}{16}$.
\end{proof}

Finally, in Section~\ref{sec:veryfew}, we dramatically decrease the upper bound on the function $g(k)$ such that every tree of order $n$ with $k$
leaves is $(n+g(k))$-unavoidable by showing the following.

\begin{theorem}\label{thm:veryfew}
Every oriented tree with $n$ nodes ($n \geq 2$) and $k$ leaves is $(\val)$-unavoidable.
\end{theorem}

The above results rely on the notion of local median order (see below). Since a local median order can easily be constructed in polynomial time, all our proofs can be transformed into polynomial-time algorithms for finding an arborescence or an oriented tree in a tournament of the size indicated in the statement.

\section{Definitions and preliminaries}
%%%%%%%%%%%%%%%%%%%

Notation generally follows \cite{BoMu08}. The digraphs have no parallel arcs and no loops.
We denote by $[n]$ the set of integers $\{1, \dots , n\}$.

Let $D$ be a digraph.
If $(u,v)$ is an arc, we say that $u$ {\bf dominates} $v$ and write $u \rightarrow v$. For any $W \subseteq V(D)$, we denote by $D\langle W\rangle$ the subdigraph induced by $W$ in $D$. 

 Let $v$ be a vertex of $D$. The {\bf out-neighbourhood} of $v$, denoted by $N^+_D(v)$, is the set of vertices $w$ such that $v \rightarrow w$. 
 The {\bf in-neighbourhood} of $v$, denoted by $N^-_D(v)$, is the set of vertices $w$ such that $w \rightarrow v$. 
 The {\bf out-degree} $d^+_D(x)$  (resp.~the {\bf in-degree} $d^-_D(x)$) is $|N^+_D(v)|$ (resp.~$|N^-_D(v)|$).

\medskip

Let $A$ be an oriented tree. The {\bf leaves} of $A$ are the vertices adjacent to (at most) one vertex in $D$. 
There are two kinds of leaves:  {\bf in-leaves} which have  out-degree $1$ and  in-degree $0$ and  {\bf out-leaves} which have out-degree $0$ and in-degree $1$. 
 The set of leaves (resp.~in-leaves, out-leaves) of $A$ is denoted by $L(A)$ (resp.~$L^-(A)$, $L^+(A)$). Trivially, $L(A) = L^+(A) \cup L^-(A)$.

A {\bf rooted tree} is an oriented tree with a specified vertex called the {\bf root}. If $A$ is a tree and $r$ a vertex of $A$, we denote by $(A,r)$ the tree $A$ rooted at $r$.
Let $A$ be a rooted tree with root $r$.
The {\bf father} of a node $v$ in $V(A)\setminus \{r\}$ is the node adjacent to $v$ in the unique path from $r$ to $v$ in $A$.
If $u$ is the father of $v$, then $v$ is a {\bf son} of $u$.
If $w$ is on the path from $r$ to $v$ in $A$, we say that $w$ is an {\bf ancestor} of $v$ and that $v$ is a {\bf descendant} of $w$.

For sake of clarity, the vertices of a tree are called {\bf nodes}.

\medskip

Let $\sigma=(v_1,v_2, \ldots, v_n)$ be an ordering of the vertices of $D$. An arc $v_iv_j$ is {\bf forward} (according to $\sigma$) if $i<j$ and {\bf backward} (according to $\sigma$) if $j<i$.
A {\bf median order} of $D$ is an ordering of the vertices of $D$ with the maximum number of forward arcs, or equivalently the minimum number of backward arcs.
In other words, a median order is an ordering of the vertices such that the set of backward arcs is a minimum feedback arc set.
Let us note basic properties of median orders of
tournaments whose proofs are left to the reader.

\begin{lemma}\label{lem:median}
 Let $T$ be a tournament and $(v_1,v_2, \ldots, v_n)$ a
median order of $T$. Then, for any two indices $i,j$ with $1 \leq i <
j \leq n$:
\medskip
\begin{enumerate}
\item[\rm (M1)] $(v_i,v_{i+1},\ldots,v_j)$ is a median order of the
  induced subtournament $T\langle \{v_i,v_{i+1},\ldots,v_j\}\rangle$.\\
 % Moreover, for every  median order  $(w_i,w_{i+1},\ldots,w_j)$ of $T\langle \{ v_i,v_{i+1},\ldots,v_j\}\rangle$, the ordering $(v_1, \dots , v_{i-1}, w_i, \dots ,w_j, v_{j+1}, \dots , v_n)$ is also a median order of $T$. 
\item[\rm (M2)] vertex $v_i$ dominates at least half of the vertices
  $v_{i+1},v_{i+2},\ldots,v_j$, and vertex $v_j$ is dominated by at least half of the vertices $v_i,v_{i+1},\ldots,v_{j-1}$.  In particular, each vertex $v_i$, $1 \leq i <n$, dominates its successor $v_{i+1}$. 
% \item[\rm (M3)] If $i<j$ and $v_{j}$ dominates $v_i$, then there exists $k$, $i< k <j$ such that $v_i\ra v_k\ra v_j$.
%\item[\rm (M4)] If $v_{i+2}$ dominates $v_i$, then $(v_1, \dots , v_{i-1}, v_{i+2}, v_i, v_{i+1}, v_{i+3}, \dots ,v_n)$ and $(v_1, \dots , v_{i-1}, v_{i+1}, v_{i+2}, v_{i}, v_{i+3}, \dots , v_n)$ are also median orders and $\{v_{i}, v_{i+1}, v_{i+2}\}$ is dominated by $v_{i-1}$ and dominates $v_{i+3}$.
\end{enumerate}
\end{lemma}

A {\bf local median order} is an ordering of the vertices of $D$ that satisfies property (M2).

Let $\sigma=(v_1, \ldots, v_m)$ be a local median order of a tournament $T$.
Let $\phi$ be an embedding  of a tree $A$ in $T$. It is {\bf $\sigma$-forward} if for all terminal interval $I = \{v_i,\ldots, v_m\}$, 
$|\phi(A) \cap I| < \frac12|I| + 1$; it is {\bf $\sigma$-backward} if for all initial interval $I = \{v_1,\ldots, v_i\}$, 
$|\phi(A) \cap I| < \frac12|I| + 1$; it is {\bf $\sigma$-nice} if it is is both a $\sigma$-forward and a $\sigma$-backward.
For all $i$ and $j$ in $\{1,\dots,m\}$ with $i<j$ and $\sigma' = (v_i,\ldots,v_j)$, if $\phi'$ is a $\sigma'$-forward, $\sigma'$-backward or $\sigma'$-nice embedding of a tree $A'$ into $T' = T\langle\{v_i,\ldots,v_j\}\rangle$, we also call it respectively a $\sigma$-forward, $\sigma$-backward, or $\sigma$-nice embedding of $A'$ into $T'$.

Havet and Thomass\'e proved the following lemma which, with an easy induction, implies that every out-arborescence of order $n$ is $(2n-2)$-unavoidable.
\begin{lemma}[Havet and Thomass\'e~\cite{HaTh00b}]\label{lem:up-embed}
Let $A$ be a tree with an out-leaf $a$. Let $T$ be a tournament and let $\sigma=(v_1, \dots , v_p)$ be a local median order of $T$.
%Let $\sigma'=(v_1, \dots , v_{p-2})$.
Every $\sigma$-forward embedding of $A-a$ in $T-\{v_{p-1}, v_p\}$ can be extended to a $\sigma$-forward embedding of $A$ in $T$.
\end{lemma}

Because we shall employ the idea used to prove it, we give the proof of Lemma~\ref{lem:up-embed}.
\begin{proof}
Assume there exists a $\sigma$-forward embedding $\phi$ of $A-a$ in $T-\{v_{p-1}, v_p\}$.
Let $b$ be the in-neighbour of $a$ in $A$, and let $v_i=\phi(b)$.
Since $\phi$ is $\sigma$-forward, $|\phi(A-a) \cap \{v_{i+1}, \dots , v_{p-2}\} |< \frac12 ({p-2-i}) + 1 = \frac{1}{2}(p-i)$.
Now, by (M2), $v_i$ has at least $\frac12(p-i)$ out-neighbours in  $\{v_{i+1}, \dots , v_{p}\}$.
Hence, $v_i$ has an out-neighbour $v_j$ in  $\{v_{i+1}, \dots , v_{p}\}\setminus \phi(A-a)$.
Set $\phi(a)=v_j$. One easily checks that $\phi$ is a $\sigma$-forward embedding of $A$ in $T$.  
\end{proof}

%Havet and Thomass\'e  also proved the following lemma which implies that every out-tree of order $n$ is $(4n-4)$-unavoidable.
%\begin{lemma}[Havet and Thomass\'e~\cite{HaTh00b}]\label{lem:nice-embed}
%Let $A$ be a tree and let $a$ be a leaf of $A$. Let $T$ be a tournament and let $\sigma=(v_{-m}, \dots , v_p)$ be a local median order of $T$.
%Let $\sigma'=(v_{-m+2},, \dots , v_{n-2})$.
%Every $\sigma$-nice embedding of $A-a$ in $T-\{v_{-m}, v_{-m+1}, v_{p-1}, v_p\}$ can be extended to a $\sigma$-nice embedding of $A$ in $T$.
%\end{lemma}

%\begin{corollary}\label{cor:4n}
%Let $A$ be a tree of order $n$, let $T$ a tree be a tree of order $4n-3$ and let $\sigma=(v_{-2n+2}, \dots , v_{2n-2})$ be a local median order of $T$.
%For any node $a\in A$, there a $\sigma$-nice embedding $\phi$ of $A$ in $T$ such that $\phi(a)=v_0$.
%\end{corollary}

%Let $f$ be a non-negative integer and let $\sigma=(v_1, \ldots, v_m)$ be a local median order of a tournament $T$.
%An embedding $\phi$ of a tree $A$ in $T$ is {\bf $\sigma$-$f$-nice} if for all terminal interval $I = \{v_i,\ldots, v_m\}$, 
%$|\phi(A) \cap I| < \frac12|I| -f +1$ and for all initial interval $I = \{v_1,\ldots, v_i\}$, 
%$|\phi(A) \cap I| < \frac12|I| -f +1$.

Let $\sigma=(v_1, \ldots, v_m)$ be a local median order of a tournament $T$.
Let $F$ be a set of vertices of $T$.
An embedding $\phi$ of a tree $A$ in $T$ is {\bf $\sigma$-$F$-nice} if for all terminal interval $I = \{v_i,\ldots, v_m\}$,
$|\phi(A) \cap I| < \frac12|I| - |F\cap I| +1$ and for all initial interval $I = \{v_1,\ldots, v_i\}$, 
$|\phi(A) \cap I| < \frac12|I| - |F\cap I| +1$.

\begin{lemma}\label{lem:islands}
Let $f$ be a positive integer and $A$ be a tree of order $n$ with root $r$.
Let $T$ be tournament of order $4n + 4f-3$ with a set $F$ of at most $f$ vertices
and let $(v_{-2n-2f +1}, \dots ,v_{2n+2f -1} )$ be a local median order of $T$ such that $v_0 \notin F$.
There is a $\sigma$-$F$-nice embedding $\phi$ of $A$ in $T$ such that $\phi(r) =v_0$, and such that for all $a \in V(A)$, $\phi(a) \notin F$.
\end{lemma}
\begin{proof}
We prove by induction on $n$, the result holding trivially when $n=1$.

Assume now that $n\geq 2$.
Let $a$ be a leaf of $A$. By directional duality, we may assume that $a$ is an out-leaf.
Let $b$ be the in-neighbour of $a$ in $A$. 
Set $p = 2n + 2f - 1$. Set $p'$ to be the smallest integer such that $p'=2n+2|F \cap (v_{-p'}, \dots ,v_{p'} )| -3$. Note that $p'$ can be obtained by starting with $p' = 2n-3$ and repeatidly replacing the value of $p'$ by the value of $2n+2|F \cap (v_{-p'}, \dots ,v_{p'} )| -3$ until the value of $p'$ remains stable. This process makes $p'$ increase at each step, and since we always have $p' \le p$ (since $|F| = f$), the process terminates. Note that $p = p' + 2 + 2|F \cap (v_{-p}, \dots ,v_{-p'-1} )| + 2|F \cap (v_{p'+1}, \dots ,v_{p} )|$.
Let $T'= T\langle \{v_{-p'}, \dots, v_{p'} \}\rangle$.
By definition, $\sigma'= (v_{-p'}, \dots ,v_{p'} )$ is a local median order of $T'$. Let $F' = F\cap V(T')$.

By the induction hypothesis, there exists a $\sigma'$-$F'$-nice embedding $\phi$ of $A-a$ in $T'$ such that for all $a \in V(A-a)$, $\phi(a) \notin F'$.
 Let $v_i=\phi(b)$.
Since $\phi$ is $\sigma'$-nice, $|\phi(A-a) \cap \{v_{i+1}, \dots , v_{p'}\} |< \frac12 ({p'-i}) -|F \cap \{v_{i+1}, \dots, v_{p'}\}| + 1 = \frac{1}{2}(p-i)- |F \cap (v_{-p}, \dots ,v_{-p'-1} )| - |F \cap (v_{p'+1}, \dots ,v_{p} )| - 1  -|F \cap \{v_{i+1}, \dots, v_{p'}\}| + 1  \le \frac{1}{2}(p-i) - |F \cap \{v_{i+1}, \dots, v_{p}\}|$. 
Now, by (M2), $v_i$ has at least $\frac12(p-i)$ out-neighbours in  $\{v_{i+1}, \dots , v_{p}\}$, so at least $\frac12(p-i)-|F \cap \{v_{i+1}, \dots, v_{p}\}|$  out-neighbours in  $\{v_{i+1}, \dots , v_{p}\}\setminus F$.
Hence, $v_i$ has an out-neighbour $v_j$ in  $\{v_{i+1}, \dots , v_{p}\}\setminus (\phi(A-a)\cup F)$.
Set $\phi(a)=v_j$. One easily checks that $\phi$ is a $\sigma$-$F$-nice embedding of $A$ in $T$.
\end{proof}

\section{Unavoidability of arborescences}\label{sec:arbo}
%%%%%%%%%%%%%%%%%%%%

The aim of this section is to prove Theorem~\ref{thm:arbo}.
We prove the following theorem which implies it directly by directional duality.

%\begin{theorem} \label{t_arbo}
%Every out-arborescence of order $n$ with $k$ out-leaves is $(n+k-1)$-unavoidable.
%\end{theorem}
\begin{theorem} \label{t_arbo}
Let $A$ be an out-arborescence with $n$ nodes, $k$ out-leaves and root $r$, let $T$ be a tournament on $m=n+k-1$ vertices,
and let  $\sigma=(v_1,v_2, \ldots, v_m)$ be a local median order of $T$.
There is an embedding $\phi$ of $A$ in $T$ such that  $\phi(r)=v_1$.
\end{theorem}
\begin{proof}
Let us describe a greedy procedure giving an embedding $\phi$ of $A$ into $T$. For each node $a$ of $A$, we fix an ordering ${\cal O}_a$ of the sons of $A$. %At each step, a node of $A$ is {\bf embedded} if it already has an image by $\phi$, and {\bf unembedded} otherwise. 
If a vertex $v_j$ of $T$ is the image of a node,  we say that it is {\bf hit} and denote its pre-image by $a_j$; in symbols $a_j=\phi^{-1}(v_j)$.

\begin{itemize}
\item Set $\phi(r)=v_1$. 
\item For $i=1$ to $m$, do 
\begin{itemize} 
\item if $v_i$ is not hit, then skip;  we say that $v_i$ is {\bf failed}.
\item if $v_i$ is hit, then assign the $|N^+(a_i)|$ first not yet hit out-neighbours of $v_i$ in $\{v_{i+1}, \dots , v_m\}$ to the sons of $a_i$ (in increasing order according to ${\cal O}_a$).
\end{itemize}

\end{itemize}

\medskip

Assume for a contradiction that this procedure does not yield an embedding of $A$ into $T$.
Then the set $F$ of failed vertices has cardinality at least $k$.
Let $B$  be the set of embedded nodes at the end of the procedure.  Since we only embed a node after its father, $A\langle B\rangle$ is an out-arborescence.
Let $L$ be the set of out-leaves of $A$ that are in $B$. Since $A\langle B\rangle$ is a sub-arborescence of $A$, we have $|L| \leq k-1$.

A node $a$ is said to be {\bf active for $i$} if $\phi(a) \in \{v_1, \dots , v_i\}$ and it has a son $b$ that is not embedded in $\{v_1, \dots , v_i\}$ (i.e. either $b$ is not embedded or $\phi(b)\in \{v_{i+1}, \dots , v_m\}$).

 Consider a vertex $v_i$ in $F$.  There is an active node for $i$, for otherwise all nodes of $A$ would be embedded (in $\{v_1, \dots , v_i\}$).
Let $\ell_i$ be the largest index such that $a_{\ell_i}$ is active for $i$. Note that by definition of active node $\ell_i< i$. Set $I_i=\{v_j \mid \ell_i < j \leq i\}$.

\begin{claim}\label{claim:leaves}
 If $v_i \in F$, then $|I_i \cap  F| \le |I_i \cap \phi(L)|$.
 \end{claim}
 \begin{subproof}
Each out-neighbour of $v_{\ell_i}$ in $I_i$ is hit for otherwise  the procedure would have assigned a son of $a_{\ell_i}$  to it. Thus $I_i\cap F\subseteq I_i  \cap N^-(v_{\ell_i})$ and so  
\begin{equation}\label{eq:1}
|I_i\cap F|\leq |I_i  \cap N^-(v_{\ell_i})| .
\end{equation} 

Let $v_j$ be a hit vertex in $I_i$.
By definition of $\ell_i$, $a_j$ is not active for $i$, so its sons (if any) are embedded in $\{v_{j+1}, \dots , v_{i-1}\}\subseteq I_i$.
Again, by definition of $\ell_i$, all the sons of $a_j$ are not active, and so their sons (if any) are embedded in $I_i$.
And so on, all descendants of $a_j$ are embedded in $I_i$ and not active. We associate to $v_j$ an out-leaf $w_j$ of $A$ which is a descendant of $a_j$. We just showed that $\phi(w_j)\in I_i$.

Consider now the vertices of $J=I_i \cap N^+(v_{\ell_i})$. As seen above, they are hit, and the descendants of their pre-images are also embedded in $I_i$.
Moreover, for each $v_j\in J$, the father of $a_j$ is embedded in $\{v_1, \dots , v_{\ell_i}\}$ for otherwise, at Step $j$, the procedure would have assigned $v_j$ to an out-neighbbour of $a_{\ell_i}$ or another active node for $i$. Hence no vertex of $J$ is the image of an ancestor of another node embedded in $J$. Consequently, the out-leaves embedded in $J$ are all distinct.
Thus 
\begin{equation}\label{eq:2}
|I_i \cap N^+(v_{\ell_i})| \leq  |I_i \cap \phi(L)| .
\end{equation}

Now, by (M2),  $|I_i  \cap N^-(v_{\ell_i})| \leq |I_i \cap N^+(v_{\ell_i})|$. Together with Equations~\eqref{eq:1} and~\eqref{eq:2}, this proves the claim. 
\end{subproof}

\begin{claim}\label{claim:inclusion}
 If $v_i \in F$ and $v_j \in F$, then either $I_i \cap I_j = \emptyset$, or $I_i \subseteq I_j$, or $I_j \subseteq I_i$. 
\end{claim}
\begin{subproof}
Let $v_i, v_j \in F$ with $i<j$. Assume for a contradiction that $I_i \cap I_j \neq \emptyset$, $I_i \not\subseteq I_j$, and $I_j \not\subseteq I_i$. 
Then $\ell_i < \ell_j < i$.
By definition of $\ell_i$, $a_{\ell_j}$ is not active for $i$. Thus all its sons are embedded in $\{v_1, \dots , v_i\}$.
Since $\{v_1, \dots , v_i\}\subseteq \{v_1, \dots , v_j\}$, $a_{\ell_j}$ is not active for $j$, a contradiction to the definition of $\ell_j$.
\end{subproof}

Now let $M$ be the set of indices $i$ such that $v_i\in F$ and $I_i$ is maximal for inclusion.
Since $v_i\in I_i$ for all $v_i\in F$, we have $F\subseteq \bigcup_{i\in M} I_i$. Moreover, by Claim~\ref{claim:inclusion}, the $I_i$, $i\in M$, are pairwise disjoint.
So $|F|= \sum_{i\in M} |I_i \cap F|$.
By Claim~\ref{claim:leaves}, we obtain $$|F|= \sum_{i\in M} |I_i \cap F| \leq \sum_{i\in M} |I_i \cap \phi(L)| \leq |\phi(L)| =|L| \leq k-1 ,$$
a contradiction. This completes the proof.
\end{proof}

\begin{obs} \label{obs:bij}
With the embedding $\phi$ constructed in the above proof, there is an injection from the set $F$ of failed vertices into $L^+(A)$ such that every failed vertex $v_i$ is mapped to an out-leaf whose image precedes $v_i$ in $\sigma$.
\end{obs}
\begin{proof}
We map the vertices $v_i$ of $F$ to an out-leaf in increasing order according to $\sigma$.

If there is an active vertex for $i$, then by Claim~\ref{claim:leaves}, $|I_i \cap  F| \le |I_i \cap \phi(L)|$. Hence, there is an out-leaf $f(v_i)$ of $A$ with image in $I_i$ (and thus preceding $v_i$ in $\sigma$) that was not assigned earlier to a failed vertex.

If there is no active vertex for $i$, then all nodes of $A$ are embedded (in vertices preceding $v_i$ in $\sigma$). Since $|F| \leq k-1$, there exists an out-leaf $f(v_i)$ which is not yet assigned to any failed vertex. Necessarily, $f(v_i)$ is embedded in a vertex preceding $v_i$ in $\sigma$.
\end{proof}

A {\bf bi-arborescence} is a rooted tree $A$ that is the union of an in-arborescence and an out-arborescence that are disjoint except in their common root, which is also the root of $A$. Theorem~\ref{t_arbo} directly implies the following corollary.

\begin{corollary}\label{cor:bi-arbo}
Let $A$ be a bi-arborescence of order $n$ with $k$ leaves.
If $A$ has at least one in-leaf and at least one out-leaf, then  $A$ is $(n+k-2)$-unavoidable.
Otherwise $A$ is $(n+k-1)$-unavoidable.
\end{corollary}

\section{Unavoidability of trees with few leaves}\label{sec:few}
%%%%%%%%%%%%%%%%%

For any rooted tree $A$ with root $r$, we partition the arcs into the {\bf upward arcs} (the ones directed away from the root) and the {
\bf downward arcs} (the ones directed towards the root). The subdigraph composed only of the upward arcs and the nodes that are in an upward arc is called the {\bf upward forest}, and the subdigraph composed only of the downward arcs and the nodes that are in a downward arc is called the {\bf downward forest}. 
The set of components of the upward (resp.~downward) forest is denoted by  ${\cal C}^{\uparrow}_r(A)$ (resp.~${\cal C}^{\downarrow}_r(A)$), or simply  ${\cal C}^{\uparrow}_r$ (resp.~${\cal C}^{\downarrow}_r$) when $A$ is clear from the context.
 Set $\gamma_r^{\uparrow} =\sum_{C\in {\cal C}^{\uparrow}_r}(|V(C)| + |L^+(C)| - 2)$ and $\gamma_r^{\downarrow} = \sum_{C\in {\cal C}^{\downarrow}_r} (|V(C)| + |L^-(C)| - 2)$. 
 Observe that each component of the upward (resp.~downward) forest contains an arc and thus at least two vertices and one out-leaf (resp.~in-leaf).
 Hence $|V(C)| + |L^+(C)| - 2 > 0$ for all $C\in {\cal C}^{\uparrow}_r$ and $|V(C)| + |L^-(C)| - 2>0$ for all $C\in {\cal C}^{\downarrow}_r$.

\begin{lemma} \label{lem:t_tree}
Let $A$ be a rooted tree with $n$ nodes and $k$ leaves such that the root $r$ of $A$ has in-degree $0$.  Then $A$ is $(n+k-1 + \gamma_r^{\downarrow})$-unavoidable.
\end{lemma}

\begin{proof}
Let $C_1$, \dots , $C_j$ be the components of the downward forest of $A$, and for $1\leq i\leq j$, let $n_i$ be the number of nodes and $k_i$ the number of in-leaves of $C_i$. 
By definition, $\gamma_r^{\downarrow} = \sum_{i =1}^j (n_i + k_i - 2)$.

Let $T$ be a tournament on $n+k-1 + \gamma_r^{\downarrow}$ vertices.

We shall use the greedy procedure described in the proof of Theorem~\ref{t_arbo}.
Observe that in this procedure, we do not need to fix the order ${\cal O}_a$ before the set of images of the sons of $a$ is known. Thus we can effectively choose which son of $a$ is embedded to which vertex with the knowledge of the set of the images of the sons of $a$.

Now we build an arborescence $A'$ from the rooted tree $A$, which we call the {\bf equivalent arborescence} of $A$.
 For all $i \in \{1,...,j\}$, do the following. Let $f_i$ be the father of the root of $C_i$. Note that $f_i$ exists since the root of $A$ has in-degree $0$, and thus is not in the downward forest. Remove all the arcs of $C_i$, add a set $N_i$ of $k_i - 1$ new nodes, and put an arc from $f_i$ to each new node and to each node of $C_i$ (except to the root of $C_i$, since that arc already exists).

Observe that $A'$ is a rooted tree with the same root as $A$ and since we removed the downward arcs and added only upward arcs, $A'$ is even an out-arborescence. By construction
$A'$ has $n+ \sum_{i=1}^j (k_i - 1)$ nodes. Let $i \in \{1,...,j\}$. The nodes of $C_i$ that are tail of an upward arc in $A$ are tail of the same upward arc in $A'$, thus they are not leaves in $A'$. Hence, each in-leaf of $C_i$ either is an in-leaf in $A$ (if it is the tail of no upward arc), or is not an out-leaf in $A'$. Therefore, in $C_i$, there are at most $n_i - k_i$ out-leaves of $A'$ that are not in-leaves in $A$. Recall that the new nodes are also out-leaves.
Therefore $A'$ has at most $k+ \sum_{i =1}^j(n_i - 1)$ out-leaves. 

Therefore by Theorem~\ref{t_arbo}, there is an embedding $\phi$ of $A'$ into $T$. We build it according to the procedure presented in the beginning of this proof.  Let $i \in \{1,...,j\}$, and consider $S_i=V(C_i)\cup N_i$. This is a set of $n_i + k_i - 1$ sons of $f_i$ in $A'$. As argued previously, we can know $\phi(S_i)$ before we choose which node of $S_i$ is embedded to which vertex. By Theorem~\ref{t_arbo}, there is an embedding $\phi_i$ from $C_i$ into $T\langle \phi(S_i)\rangle$. Now for each node $a$ in $C_i$, we choose $\phi_i(a)$ as its image by $\phi$.

Consider now $\psi$ the restriction of the resulting embedding $\phi$ to $V(A)$. For all $i \in \{1,...,j\}$, $\psi$ coincides with $\phi_i$. Hence $\psi$ preserves the upward arcs since all the upward arcs of $A$ are in $A'$, and preserves the downward arcs since each downward arc of $A$ is in some $C_i$. Therefore $\psi$ is an embedding of $A$ into $T$.
\end{proof}

\medskip

We are now able to prove Theorem~\ref{t_tree2} which states that every oriented tree of order $n$ with  $k$ leaves is $(\frac32n+\frac{3}{2}k-2)$-unavoidable.
\begin{proof}[Proof of Theorem~\ref{t_tree2}]
Let $T$ be a tournament on $\frac32(n+k)-2$ vertices.
Let $A$ be an oriented tree with $n$ nodes an $k$ leaves. 
Pick a root $r$ such that $\min(\gamma_r^{\uparrow},\gamma_r^{\downarrow})$ is minimum. By directional duality, we may assume that this minimum is attained by $\gamma_r^{\downarrow}$. 

Since $\gamma_r^{\downarrow} \le \gamma_r^{\uparrow}$, we have $\gamma_r^{\downarrow} \le \frac12 (\gamma_r^{\uparrow}+\gamma_r^{\downarrow})$. 
For a rooted tree $A$, let $L'(A)$ be the set of leaves of $A$ distinct from the root.
Note that if $A_1$ and $A_2$ are two rooted trees that are disjoint except in one vertex which is the root of $A_2$, then 
$A_1\cup A_2$ is a tree, and if we root it at the root of $A_1$, then $|L'(A_1\cup A_2)| \geq |L'(A_1)| + |L'(A_2)|-1$.  By applying that successively for all the components of the upward and downward forests of $(A,r)$, we get that $\gamma_r^{\uparrow}+\gamma_r^{\downarrow} \le n+k-2$, and thus $\gamma_r^{\downarrow} \le \frac12(n+k) - 1$. Hence, $T$ has at least $n+k-1 + \gamma_r^{\downarrow}$ vertices.

Suppose for a contradiction that $r$ has an in-neighbour $s$. The downward forest $F_s$ of $(A,s)$ is obtained from the downward forest $F_r$ of  $(A,r)$ by removing the arc $sr$ and possibly $s$ or $r$ if they become isolated.
All components of $F_r$ not containing $sr$ are also components of $F_s$ and the component $C_0$ of $F_r$ containing $sr$ either disappears (when $sr$ is the sole arc of $C_0$), or loses one vertex (when $r$ or $s$ is a leaf of $C_0$), or is split into two components having in total as many vertices as $C_0$ and at most one more in-leaf than $C_0$. In any case,  $\gamma_s^{\downarrow} < \gamma_r^{\downarrow}$, a contradiction.

Consequently $r$ has in-degree $0$. Lemma~\ref{lem:t_tree} finishes the proof.
\end{proof}

\section{Unavoidability of trees with many leaves}\label{sec:many}
%%%%%%%%%%%%%%%%%%%

The aim of this section is to establish Theorem~\ref{t_leaves}, which we recall.\\

\noindent{\bf Theorem~\ref{t_leaves}.} 
Every oriented tree with $n\ge 3$ vertices and $k$ leaves is $(\frac92n - \frac52k - \frac92)$-unavoidable.

\begin{proof}
Set $m = \lceil\frac92n - \frac52k - \frac92\rceil$.
Let $T$ be a tournament on $m$ vertices.
Let $A$ be an oriented tree with $n$ nodes and $k$ leaves. 
If $A$ is a bi-arborescence, then we have the result by Corollary~\ref{cor:bi-arbo}. Henceforth, we assume that $A$ is not a bi-arborescence.
In particular, $k<n-1$.

The {\bf out-leaf cluster} of $A$, denoted by $S^+$, is the set of nodes of $A$ defined recursively as follows. Each out-leaf $A$ is in $S^+$; if $a$ is a node with exactly one in-neighbour and all its out-neighbours are in $S^+$, then $a$ is also in $S^+$. We similarly define the {\bf in-leaf cluster} $S^-$ of $A$.  Note that $A\langle S^-\rangle$ is a forest of in-arborescences, and $A\langle S^+\rangle$ is a forest of out-arborescences.
Moreover, $S^-\cap S^+=\emptyset$ because $A$ is not a bi-arborescence.

The {\bf heart} of $A$, denoted by $H$, is the tree $A -(S^-\cup S^+)$. Set $n_H=|V(H)|$ and $k_H=|L(H)|$.
We first note that each out-leaf of $H$ has a neighbour in $S^-$, since otherwise it would be in $S^+$. Similarly, each in-leaf of $H$ has a neighbour in $S^+$. In particular, $|S^-| \ge |L^+(H)|$ and $|S^+| \ge |L^-(H)|$.

We now describe an algorithm yielding an embedding $\phi$ of the tree $A$ into $T$.
It proceeds in three phases: in the first phase, we embed the heart of $A$, in the second phase we embed the out-leaf cluster, and in the third phase we embed the in-leaf cluster.
At each step, a node of $A$ is {\bf embedded} if it already has an image by $\phi$, and {\bf unembedded} otherwise. If a vertex $v_j$ of $T$ is the image of a node,  we denote this node by $a_j$; in symbols $a_j=\phi^{-1}(v_j)$. We say that a vertex is {\bf hit} if it is the image of a node.

\medskip
Let  $\sigma=(v_1,v_2, \ldots, v_m)$ be a local median order of $T$. Our algorithm heavily relies on $\sigma$ to embed $A$ in $T$.
It distinguishes two cases. We first deal with the easier case when one of $S^-, S^+$ is empty.
By directional duality, we may assume that  $S^-=\emptyset$. In that case, we proceed only in two phases. (A third phase to embed $S^-$ is useless.)
First, by Lemma~\ref{lem:islands}, we find a $\sigma$-nice embedding of $H$ in $T\langle\{v_1, \dots , v_{4n_H-3}\})$. Then, by repeated application of Lemma~\ref{lem:up-embed}, one can find an embedding of $A$ into $T\langle\{v_1, \dots , v_{4n_H+2|S^+|-3}\})$.
We have $4n_H+2|S^+|-3 = 4n - 2|S^+| -3 \leq 4n -2k -3 \leq m$ (since $k<n-1$).

%\fh{{\bf Remarques:} 1) on est oblige d'utiliser $4n_H-3$ et pas $4n_H-4$ car on peut avoir $n_h=1$. Non! on a forcément $n_h\ge2$!
%2) en faisant comme ca en rajoutant $4n-3$ sommets a gauche pour etre sur d'avoir un backward embedding apres Phase 2, on obtient la borne $8n -6k -6$.}

\medskip

Let us now deal with the more complicated case when both $S^-$ and  $S^+$ are non-empty.
We first have to find an adequate root of $A$ to start.

Let ${\cal C}^{\downarrow}_r={\cal C}^{\downarrow}_r(H)$ and ${\cal C}^{\uparrow}_r={\cal C}^{\uparrow}_r(H)$. For a root $r$ of $H$, let $\beta^{\downarrow}_r=\sum_{C\in {\cal C}^{\downarrow}_r} (3|V(C)|-3) +  2|L^-(H)|$ and $\beta^{\uparrow}_r=\sum_{C\in {\cal C}^{\uparrow}_r} (3|V(C)|-3) +  2|L^+(H)|$. Let $r$ be a root that minimizes $\min\{\beta^{\downarrow}_r, \beta^{\uparrow}_r\}$. By directional duality, we may assume that $\beta^{\downarrow}_r = \min\{\beta^{\downarrow}_r, \beta^{\uparrow}_r\}$.
Therefore, $\beta^{\downarrow}_r \leq \frac{1}{2}\beta^{\downarrow}_r + \frac{1}{2}\beta^{\uparrow}_r = \frac{3}{2}n_H + k_H - \frac32$. We can assume that $r$ has in-degree $0$, since each in-neighbour $s$ of $r$ satisfies  $\beta^{\downarrow}_s \leq  \beta^{\downarrow}_r$.

Let us now detail our algorithm.
 Let $\ell = n_H - k_H -1+  \sum_{C\in {\cal C}^{\downarrow}_r} (|V(C)|-1) +  2|L^-(H)| + 2 |S^-|$. Note that $\ell \geq 1$ because $|S^-|\geq 1$.
 Let $p = \ell +n_H+k_H-1 + \gamma^{\downarrow}_r$.
 \begin{itemize}
 \item[] \underline{Phase 1}: We embed $H$ in $T\langle \{v_{\ell+1},\dots,v_{p}\}\rangle$ using the procedure of Lemma~\ref{lem:t_tree} for $H$. 
Note that in this procedure, we embed the equivalent arborescence $H'$ of $H$ which is bigger than $H$. Here we keep all vertices of $H'$ embedded until the end of Phase~2. 
 
 \item[] \underline{Phase 2}: While there is an unembedded node in $S^+$, let $i$ be the smallest integer such that $\phi^{-1}(v_i)$ has an unembedded out-neighbour in $S^+$, and take the first (i.e. with lowest index) out-neighbour of $v_i$ in $\{v_{i+1}, \dots , v_m\}$ that is not yet hit and assign it to an unembedded out-neighbour in $S^+$. 
 
 Unembed all vertices of $H'-H$.
 
 \item[] \underline{Phase 3}: While there is an unembedded node in $S^-$, let $i$ be the largest integer such that $\phi^{-1}(v_i)$ has an unembedded in-neighbour in $S^-$, and take the last (i.e. with highest index) in-neighbour of $v_i$ in $\{v_{1}, \dots , v_{i-1}\}$ that is not yet hit and assign it to an unembedded in-neighbour in $S^-$. 
\end{itemize}

\medskip

Let us prove that this algorithm embeds all nodes of $A$.
First, by Lemma~\ref{lem:t_tree}, all vertices of $H$ are embedded in Phase 1.

\smallskip
Let us now prove that all nodes of $S^+$ are embedded in Phase 2.
Let $B$ be the subtree of $A$ induced by $V(H)\cup S^+$ and let $B'$ be the out-arborescence obtained from $B$ by replacing $H$ by the equivalent arborescence $H'$. Observe that Phase 1 and Phase 2, may be seen as embedding $B'$ and extracting a copy of $B$ from $B'$ at the same time.
Let us show that our algorithm embeds the whole $B'$ (and thus the whole $B$) in Phases 1 and 2.
 The equivalent arborescence $H'$ has $n_H + \sum_{C\in {\cal C}^{\downarrow}_r} (|L^-(C)|-1)$ nodes. Thus $B'$ has $n_H + \sum_{C\in {\cal C}^{\downarrow}_r} (|L^-(C)|-1) + |S^+|$ nodes. 

All the leaves of $A$ are either in $S^-$ or in $S^+$, thus $k \le |S^-|+|S^+|$.
Therefore $$m\geq  \frac92n - \frac52k  - \frac{9}2 \geq \frac{9}{2}n_H +2|S^-|+2|S^+| - \frac{9}2 .$$
 For all $C \in  {\cal C}^{\downarrow}_r$, we have $|L^-(C)| \le |V(C)|$, thus 
 $$\sum_{C\in {\cal C}^{\downarrow}_r} (|V(C)| + 2|L^-(C)|-3) +  2|L^-(H)| \le \sum_{C\in {\cal C}^{\downarrow}_r} (3|V(C)|-3) +  2|L^-(H)| \le \frac{3n_H}2 + k_H - \frac32.$$ 
 The two previous equations yield
$$ m  \geq  3n_H - k_H - 3 +\sum_{C\in {\cal C}^{\downarrow}_r} (|V(C)| + 2|L^-(C)|-3) +  2|L^-(H)| +2|S^-|+2|S^+|~,$$
so $\displaystyle  m-\ell \geq   2n_H  + 2\sum_{C\in {\cal C}^{\downarrow}_r} (|L^-(C)|-1) + 2|S^+| - 2 \geq 2|B'|-2$.
  
Each time we embed a node of $B'$ during Phases 1 and 2, our procedure takes the first (i.e. with lowest index) out-neighbour of a vertex $v_i$ in $\{v_{i+1}, \dots , v_m\}$ that is not yet hit and assigns it to an unembedded out-neighbour of $\phi^{-1}(v_i)$. Therefore, at each step, $\phi$ is a $\sigma$-forward embedding of $B''$, the so far constructed sub-out-arborescence of $B'$, into $T\langle \{v_{\ell+1},...,v_{\ell+2|B''| -2}\rangle$.
Thus, as in Lemma~\ref{lem:up-embed}, every vertex $v_i$ has an out-neighbour in $\{v_{i+1},...,v_{\ell+2|B''|}\} \setminus \phi(B'')$ and the procedure can continue. Hence, $B'$ can be embedded into $T\langle \{v_{\ell+1},...,v_m\}\rangle$.
 
\medskip

Assume for a contradiction that the algorithm fails in Phase 3, which means a node $a$ in $S^-$ is not embedded.
 We can choose such a node $a$ whose out-neighbour $b$ is embedded. Let $v_i$ be the image of $b$. 
 Observe that $b$ is in $S^-\cup V(H)$, so it has been embedded in Phase 1 or 3, and necessarily must be in $\{v_1, \dots , v_{p}\}$.

Consider the moment when we try to embed the in-neighbours of $b$ during Phase 3.
Let $\hit$ be the number of vertices of $\{v_1, \dots , v_{i-1}\}$ that are hit at this moment.
Since $a$ is not embedded, we have $\hit \geq |N^-(v_i) \cap \{v_1, \dots , v_{i-1}\}| - |N^-_A(b)| +1$.
By (M2),  $|N^-(v_i) \cap \{v_1, \dots , v_{i-1}\}| \geq \frac{i-1}{2}$. So
\begin{equation}\label{eq:hit}
\hit  \geq  \frac{i-1}{2} - |N^-_A(b)| +1.
\end{equation}

Let us give some upper bounds on $\hit$. Let $O_{<i}$ be the set of out-leaves of $H$ embedded at some $v_j$ with $\ell+1\leq j < i$ and let $O_{\ge i}$ be the set of out-leaves of $H$ embedded at some $v_j$ with $i\leq j \leq p$. 
We have $\hit = \hit_{2} + \hit_{3}$, where $\hit_{2}$ (resp.~$\hit_3$) is the number of vertices of $\{v_1, \dots , v_{i-1}\}$ that are hit in Phase 1 and 2 (resp.~Phase 3 until the considered moment). 

At the considered moment, the algorithm has yet to embed the in-neighbours of $b$ and the in-neighbours in $S^-$ of the nodes embedded at each $v_j$ for $j < i$. 
  As noted previously, each out-leaf of $H$ has an in-neighbour in $S^-$. Therefore, each out-leaf of $O_{< i}$ has an in-neighbour in $S^-$ that is not yet embedded. Hence 
  \begin{equation}\label{eq:hit3}
  \hit_3 \leq |S^-| - |O_{<i}| - |N^-_A(b)|.
 \end{equation}

Consider now the embedding of $H$. It is made using the procedure of Lemma~\ref{lem:t_tree}, which applies the procedure of Theorem~\ref{t_arbo} on $H'$. Let $k_{H'}$ be the number of out-leaves of $H'$. All the out-leaves of $H$ are also out-leaves in $H'$. 
Moreover, by Observation~\ref{obs:bij}, we can map each failed vertex to an out-leaf of $H'$ whose image precedes the failed vertex in $\sigma$.
But there are $k_{H'}-1$ failed vertices, therefore each out-leaf of $O_{\ge i}$ except at most one corresponds to two vertices $v_l$ with $l \ge i$.  Thus, there are at least $2|O_{\ge i}| -1$ vertices in $\{v_i, \dots , v_{p}\}$. Moreover, $\sum_{C\in {\cal C}^{\downarrow}_r} (|L(C)|-1)$ vertices of $H'-H$ are unembedded at the end of Phase 2 and were neither out-leaves of $H$ nor failed vertices.
Hence, 
\begin{equation}\label{eq:hit2}
\hit_2 \leq p-\ell - \sum_{C\in {\cal C}^{\downarrow}_r} (|L(C)|-1) -  2|O_{\ge i}| + 1 = n_H + k_H +  \sum_{C\in {\cal C}^{\downarrow}_r} (|V(C)| - 1) - 2|O_{\ge i}| .
 \end{equation}

Since all vertices hit in Phases 1 and 2 are in $\{v_{\ell+1}, \dots , v_m\}$, we trivially have
\begin{equation}\label{eq:hit2bis}
\hit_2 \leq i -\ell -1
\end{equation}

Summing 2 Eq.~\eqref{eq:hit} + 2 Eq.~\eqref{eq:hit3} +  Eq.~\eqref{eq:hit2} + Eq.~\eqref{eq:hit2bis} yields:
\begin{equation*} 
\begin{split}
\ell &\le n_H + k_H +  \sum_{C\in {\cal C}^{\downarrow}_r} (|V(C)|-1)) - 2|O_{\ge i}|+ 2|S^-| - 2|O_{<i}| - 2 \\
&\le n_H + k_H -2 L^+(H) +  \sum_{C\in {\cal C}^{\downarrow}_r} (|V(C)|-1))+ 2|S^-| - 2 \\
&= n_H - k_H -2 +2(L^-(H)) +  \sum_{C\in {\cal C}^{\downarrow}_r} (|V(C)|-1))+ 2|S^-|\\
&= \ell - 1,
\end{split}
\end{equation*}
a contradiction.
\end{proof}

\section{Unavoidability of trees with very few leaves}\label{sec:veryfew}
%%%%%%%%%%%%%%%%%%

The aim of this section it to prove Theorem~\ref{thm:veryfew} which states that every oriented tree with $n$ nodes and $k$ leaves is $(\val)$-unavoidable. 
Since the result holds for paths, we shall only consider trees that are not paths.

Let $A$ be a tree which is not a path.
A {\bf branch-node} of $A$ is a node with degree at least $3$ and a {\bf flat node} is a node with degree $2$.
A {\bf segment} in $A$ is a subpath whose origin is a branch-node, whose terminus is either a branch-node or a leaf, and whose internal nodes are flat nodes.
If its terminus is a branch-vertex, then the segment is an {\bf inner segment}; otherwise it is an {\bf outer segment}. 
The {\bf opposite} of an inner segment $S$, denoted by $\overline{S}$, is the inner segment with origin the terminus of $S$ and terminus the origin of $S$. 

A {\bf stub} is a tree such that :
\begin{itemize}
\item[(i)] every inner segment has at most three blocks; moreover, if it has three blocks then its first and third block have length $1$, and if it has two blocks then one of them has length $1$.
%\item every outer segment has at most two blocks and if it has two, its first block has length $1$.
\item[(ii)] every outer segment has length $1$.
\end{itemize}

Our proof of  Theorem~\ref{thm:veryfew}  involves two steps.
We first prove the following lemma, which shows that it is sufficient to concentrate on stubs.

\begin{lemma}\label{lem:reduc}
If there exists a function $f$ such that every stub of order $n$ and $k\geq 6$ leaves is $(n+f(k))$-unavoidable, then
every tree of order $n$ with $k\geq 3$ leaves is $(n+\max\{f(2k-2b)+b \mid 0\leq b\leq k-3\})$-unavoidable. 
%If there exists a function $f$ such that every tree of order $n$ and $k\geq 6$ leaves is $f(k)$-unavoidable, then
%every tree of order $n$ with $k\geq 3$ leaves is $n+\max\{f(2k-2b)+b \mid 0\leq b\leq k-3\}$-unavoidable. 
\end{lemma}

We then prove the following result on the unavoidability of stubs.

\begin{lemma}\label{lem:stub}
Every stub with $n$ nodes and $k\geq 6$ leaves is $(\vstub)$-unavoidable.
\end{lemma}

Theorem~\ref{thm:veryfew} follows directly from Lemmas~\ref{lem:reduc} and \ref{lem:stub}.

\subsection{Reducing to stubs}
%%%%%%%%%%%%%%%%%%%

\subsubsection{Toolbox}
%%%%%%%%%%%%

Let $P=(x_1,\dots ,x_n)$ be a 
path. We say that $x_1$ is the {\bf origin} of $P$ and $x_n$ is the {\bf
terminus } of $P$; $x_1$ and $x_n$ are the {\bf ends} of $P$. 
If $x_1\ra x_2$, $P$ is an {\bf out-path}, 
otherwise $P$ is an {\bf in-path}. The {\bf directed out-path} of order $n$ is 
the path $P=(x_1,\dots ,x_n)$ in which $x_i\ra x_{i+1}$ for all $i \in [n-1]$;
the dual notion is that of a {\bf directed in-path}. 
The {\bf length} of a path is its number of edges.
An {\bf $\ell$-out-path} (resp.~{\bf $\ell$-in-path}) is an out-path (resp.~in-path) of length $\ell$.
We denote 
the path $(x_2,\dots ,x_n)$ by $^*P$.

The {\bf blocks} of $P$ are the maximal directed subpaths of $P$. 
We enumerate the blocks of $P$ from the origin to the terminus. The first block 
of $P$ is denoted by $B_1(P)$ and its length by  $b_1(P)$. Likewise, the $i^{\rm 
th}$ block of $P$ is denoted by $B_i(P)$ and its length  by $b_i(P)$. The path $P$ is 
totally described by the signed sequence $sgn(P)(b_1(P),b_2(P),\dots
,b_k(P))$, called its {\bf type},  
where $k$ is the number of blocks of $P$ and $sgn(P)=+$ if $P$ is an out-path 
and  
$sgn(P)=-$ if $P$ is an in-path.

\medskip

Thomason~\cite{Tho86} proved the following two theorems.  See also \cite{HaTh00} for a short proof of the first one.

\begin{theorem}[Thomason \cite{Tho86}]\label{thm:dep2}
Let $P$ be an oriented path of order $n$.
Let $T$ be a tournament of order $n+1$ and $X$ a set of $b_1(P)+1$ vertices.
There exists a copy of $P$ in $T$ with origin in $X$.
\end{theorem}

\begin{theorem}[Thomason, \cite{Tho86}]\label{thm:2bouts}
Let $P$ be a non-directed path of order $n$ with first and last block of
length $1$. Let $T$ be a tournament of order $n+2$ and $X$ and $Y$ be two
disjoint subsets of $T$ of order at least $2$. 

If $P\neq \pm (1,1,1)$, then there is a copy of $P$ in
$T$ with origin in $X$ and terminus in $Y$.
\end{theorem}

The idea to find a tree $A$ in a tournament $T$ is to break some segments $S$, that is to remove the arcs and internal vertices of some subpaths $R_S$ satisfying the hypothesis of Theorem~\ref{thm:2bouts} if $S$ is an inner segment and of Theorem~\ref{thm:dep2} if $S$ is an outer segment.
Then we find the resulting forest $A'$ in $T$. Finally, we reconstruct the broken segment using  
Theorems~\ref{thm:dep2} and \ref{thm:2bouts}.
However, those theorems prescribe the origin and terminus not in a vertex but in a set of two vertices.
Therefore, we need to have a little more than the paths of $S-R_S$ to reconstruct $S$. This is captured by the notion of fork.

The {\bf fork} $F$ of type $\tau=sgn(F)(b_1(F),b_2(F),\cdots
,b_k(F))$ is
the tree with vertex set $\{x_1, \dots ,x_{n-2}, p_1,
p_2\}$ such that $(x_1, \dots 
,x_{n-2}, p_1)$ and $(x_1, \dots ,x_{n-2}, p_2)$ are paths 
of type $\tau$.
The vertex $x_1$ is the {\bf origin} and $p_1$ and $p_2$
are the {\bf points} of the fork.
%Analogously as for paths, a fork is an {\bf out-fork} if $sgn(F)=+$ and an {\bf in-fork} otherwise.
%Moreover a fork of type $\pm (p)$ is said to be {\bf directed} and a fork of type
%$\pm (1,p)$ is said to be {\bf reversing}.

Let $P$ be a path of length at least $2$. Its {\bf stump type} is
\begin{itemize}
\item[(i)] $sgn(P)(b_1(P)-1)$ if $b_1(P)\geq 2$, and the type of $P$ is  not $+(p,1,q)$ (with $q\geq 2$),
\item[(ii)] $sgn(P)(b_1(P))$ if $P$ is of type  $+(p,1,q)$ (with $q\geq 2$),
\item[(iii)] $sgn(P) (1)$ if $b_1(P)=b_2(P)=1$ and $P$ is not of type $\pm (1,1,1,p)$ (with $p\geq 2$)
or $+(1,1,1,1,1)$,
\item[(iv)]  $sgn(P) (1,1)$ if $P$ is of type $\pm (1,1,1,p)$ (with $p\geq 2$)
or $+(1,1,1,1,1)$,
\item[(v)]  $sgn(P) (1,b_2(P)-1)$ if $b_1(P)=1$ and $b_2(P)\geq 2$.
\end{itemize}

\subsubsection{Proof of Lemma~\ref{lem:reduc}}
%%%%%%%%%%%%%%%%%%%%%%%%

Let $A$ be a tree of order $n$ with $k$ leaves.
An inner segment is {\bf unbreakable} if either it is directed, or it has two blocks at least one of which has length $1$, or it has three blocks with the first and the last of length $1$. Otherwise it is {\bf breakable}. 
We construct an oriented forest $B$ from $A$ by applying the following two operations.
\begin{itemize}
\item[1.] For each outer segment $S$ of $A$ of length at least $2$ and with origin $x$, we replace $S$ by a fork $F_S$ with origin $x$ whose type is the stump type of $S$.
The {\bf remainder} of $S$, denoted by $R_S$, is the path obtained from $S$ by removing the $|F_S|-1$ first vertices of $S$.
\item[2.] For each breakable inner segment $S$ with origin $x$ and terminus $y$, replace $S$ by a fork $F_{S}$ with origin $x$ whose type is the stump type of $S$, and a fork $F_{\overline{S}}$ with origin $y$ whose type is the stump type of $\overline{S}$. The {\bf remainder} of $S$, denoted by $R_S$, is the path obtained from $S$ by removing the $|F_S|-2$ first and the $|F_{\overline{S}}| -2$ last vertices of $S$.
\end{itemize}

Let $b$ be the number of breakable inner segments and $a$ be the number of outer segments of length $1$. Note that $0\leq b \leq k-3$. The forest $B$ has $b+1$ components.
%For all $i\in [p]$, let $n_i$ (resp.~$k_i$) be the number of nodes (resp.~leaves) of $A_i$.
By construction, $B$ has at most $2k+ 4b-a$ leaves, because Operation $1$ replace the leaf of each outer segment of length at least $2$ by the two leaves of a fork, and Operation 2 create four leaves when a breakable inner segment is broken by Operation  2.
Moreover $|B| \leq |A|+b$, because $|F_S|\leq |S|$ for each outer segment $S$ of length at least $2$, and $|F_S| + |F_{\overline{S}}| \leq |S| +1$ for every breakable inner segment.
Finally, observe that every component $C$ of $B$ has at least $6-a_C$ leaves, where $a_C$ is the number of outer segments of length $1$ of $A$ contained in $C$. Thus a component has at most $2k -2b$ leaves.

Let $A_i$,  $i\in [b+1]$, be the components of $B$. Each $A_i$ is a stub.
Thus, by hypothesis, each $A_i$ is $(|A_i| + f(2k-2b))$-unavoidable, and so $B$ is $(|B|+ f(2k-2b))$-unavoidable.
Let $T$ be a tournament of order $n+\max\{f(2k-2b)+b \mid 0\leq b\leq k-3\}$.
$T$ contains $B$ has a subdigraph.
We can now transform $B$ into $A$ as follows.
\begin{itemize}
\item[0.] Initialize $A^*$ to $B$.
\item[1.] For each breakable inner segment $S$, let $U$ be a set of $|R_S|-2$ vertices in $T-A^*$. Let $X_S$ be the set of points of $F_S$ and $X_{\overline{S}}$ be the set of points of $\overline{S}$. Note that the stub types are defined in such a way that $R_S$ is not of type $\pm (1,1,1)$.
Since $R_S$ has first and last block of length $1$,  by Theorem~\ref{thm:2bouts}, in $T\langle U\cup X_S \cup X_{\overline{S}}\rangle$, there is a copy $R_S^*$ of $R_S$ with origin in $X_S$ and terminus in $X_{\overline{S}}$.
Remove from $A^*$ the point of $F_S$ which is not the origin of $R^*_S$, and the point of $F_{\overline{S}}$ which is not the terminus of $R^*_S$; add the path $R^*_S$ to $A^*$.
\item[2.] For each outer segment $S$, let $U$ be a set of $|R_S|-1$ vertices in $T-A^*$. Let $X_S$ be the set of points of $F_S$. Since $R_S$ has first and last block of length $1$,  by Theorem~\ref{thm:dep2}, in $T\langle U\cup X_S \rangle$, there is a copy $R_S^*$ of $R_S$ with origin in $X_S$.
Remove from $A^*$ the point of $F_S$ which is not the origin of $R^*_S$, and add the path $R^*_S$ to $A^*$.
\end{itemize}
Step 1 of this procedure reconstructs the breakable inner segments (which were broken) and Step 2 completes the outer segments of length at least $2$.
Therefore at the end of the procedure $A^*$ is the desired tree $A$.
This completes the proof of Lemma~\ref{lem:reduc}.

\subsection{Unavoidability of stubs}
%%%%%%%%%%%%%%%%%%

The aim of this subsection is to prove Lemma~\ref{lem:stub}.
We need some preliminary results.

\subsubsection{Toolbox}
%%%%%%%%%%%%%

\begin{lemma}\label{lem:lapin}
Let $k$ be a positive integer.
Let $T$ be tournament of order $m\geq 4k$ and let $(v_1, \dots , v_m)$ be a local median order of $T$.
There are at least $k$ internally disjoint directed $2$-out-paths with origin $v_1$ and terminus in $\{v_{m-4k+2}, \dots , v_m\}$.
\end{lemma}
\begin{proof}
Let $S=\{m-4k+2, \dots , m-1\}$, and $M=\{v_2, \dots, v_{m-4k+1}\}$.
If $v_1$ dominates at least $k$ vertices in $S$, then the paths $(v_1, v_i, v_{i+1})$ for $v_i\in N^+(v_1)\cap S$ give the result.
Therefore, we may assume that $v_1$ dominates at most $k-1$ vertices in $S$.
Now by (M2), $v_1$ dominates at least $\frac{m}{2}-1$ vertices of $\{v_2, \dots, v_{m-1}\}$, and so at least $\frac{m}{2} -k$ in $M$.
Again by (M2), $v_{m-4k+2}$ is dominated by at least $\frac{m-4k}{2}$ vertices of $\{v_2, \dots, v_{m-4k+1}\}$.
Thus $A=N^+(v_1) \cap N^-(v_{m-4k+2})\cap M$ has cardinality at least $\frac{m}{2} -k + \frac{m-4k}{2} - (m-4k)= k$.
Hence the paths $(v_1, v_i, v_{m-4k+1})$ for $v_i\in A$ give the result.
\end{proof}

In the previous lemma, the termini of the paths may be the same (namely, $v_{m-4k+2}$). However, by applying it successively for $k'$ from $1$ to $k$, we directly obtain the stronger following lemma:

\begin{lemma}\label{lem:lapin2}
Let $k$ be a positive integer.
Let $T$ be tournament of order $m\geq 4k$ and let $(v_1, \dots , v_m)$ be a local median order of $T$.
There are at least $k$ internally disjoint directed $2$-out-paths with origin $v_1$ and distinct termini in $\{v_{m-4k+2}, \dots , v_m\}$.
\end{lemma}

\begin{lemma}\label{lem:arbos}
Let $f, p ,s$ be positive integers with $s > p$. For $q\in [p]$, let $A_q$ be an out-arborescence with $n_q$ nodes and $k_q$ out-leaves. 
Let $T$ be a tournament of order $m = s + \sum_{q\in [p]}(n_q +k_q-1) + 2f - 1$,  $\sigma=(v_1, \dots ,v_m)$ be a local median order of $T$ and $F$ a set of at most $f$ vertices of $T$. If there are indices $1\leq i_1 < \cdots < i_p \leq s$ such that $v_{i_q}\notin F$  for all $q \in [p]$, then there is an embedding $\phi$ of $B=A_1+\dots+A_p$ in $T$ such that the root of $A_q$ is embedded at $v_{i_q}$ for all $q \in [p]$,  and  $\phi(a) \notin F$ for all $a \in V(B)$.
\end{lemma}
\begin{proof}
Let us build an arborescence $A'$ on which to apply Theorem~\ref{t_arbo}. Add to $B$ a node $a$ with an out-going arc to each of the roots of the $A_q$'s, and another node $b$ with an out-going arc to $a$ and to $s - p$ new nodes, $a_1, \dots ,a_{s-p}$. Order the sons of $b$ in the order $(a,a_1,\dots,a_{s-p})$.

Now let $T'$ be the tournament obtained from $T$ by adding a transitive tournament $S$ on $s-p+2$ vertices  $v_{p-s-1}, v_{p-s}, \dots, v_0$ in the transitive order. Let all these vertices except $v_{p-s}$ dominate all vertices of $T$, and let $v_{p-s}$ dominate $\{v_{i_q} \mid q\in [p]\}\cup \{v_i \mid s+1\leq i\leq m\}$ and be dominated by the $s-p$ other vertices. 

Note that $A'$ has $n'=s -p+2+ \sum_{q\in [p]} n_q$ nodes and $k'=s -p+\sum_{q\in [p]} k_q$ out-leaves, and that $T'$ has $m+s-p+2 = n' + k' -1+2f$ vertices.
The idea is to embed $A'$ into $T'$ using Theorem~\ref{t_arbo} with the ordering $\sigma'=(v_{p-s-1}, \dots ,v_{m})$.
We thus need to show that $\sigma'$ is a local median order.

\begin{claim}
$\sigma'$ is a local median order.
\end{claim}
\begin{subproof}
We need to prove that $\sigma'$ has property $(M2)$.
 Let $i <j$ be two integers in $\{v_{p-s-1}, \dots ,v_{m}\}$.
 
 Let us first show that $v_i$ dominates at least half of the vertices $v_{i+1}, \dots , v_j$.
 If $i>0$, it follows from the fact that $\sigma$ is a local median order.
 If $i\leq 0$ and $i\neq p-s$, then $v_i$ dominates all  the vertices $v_{i+1}, \dots , v_j$ by construction.
 If $i=p-s$, it holds because $v_{p-s}$ has $s-p$ out-neighbours in $\{v_{p-s+1},\dots, v_0\}$ and at most $s-p$ in-neighbours with positive index.
 
 Let us first show that $v_j$ is dominated by at least half of the vertices $v_{i}, \dots , v_{j-1}$.
 If $j\leq 0$, then by construction it is dominated by all vertices $v_{i}, \dots , v_{j-1}$.
Assume now that $j>0$. If $i> 0$, it  follows from the fact that $\sigma$ is a local median order.
 If $i\leq 0$, then $v_j$ is dominated by at least half the vertices of $\{v_1, \dots , v_{j-1}\}$ because $\sigma$ is a local median order, it is dominated by $v_0$, and dominates at most one vertex (namely $v_{p-s}$) with non-positive index. Therefore if dominates at least  half the vertices of $\{v_{i}, \dots , v_{j-1}\}$.
 \end{subproof}

Consequently, following the procedure of Theorem~\ref{t_arbo}, we get an embedding $\phi$ of $A'$ into $T'$. This embedding may however embed some nodes at vertices of $F$.
For each vertex $r\in F$ in order, do the following: if a node $a$ of $A'$ is embedded at $r$, let $b$ be the father of $a$ in $A'$; add a leaf in $A'$ with father $b$, and put it just before $a$ in the order of the sons of $b$; reapply Theorem~\ref{t_arbo} to obtain an embedding of the new version of $A'$.
Note that in this construction, we add at most $k$ vertices and leaves to $A'$. Thus Theorem~\ref{t_arbo} is still applicable, and in the resulting embedding, the only nodes that are embedded at vertices of $F$ are not in $A$.
With the right order on the neighbours of $b$ (namely $(a,a_1,\dots,a_{s-j})$), the algorithm of Theorem~\ref{t_arbo} maps $b$ to $v_{p-s-1}$, $a$ to $v_{p-s}$,  $a_i$ to $v_{p-s+i}$ for all $i \in [s-p]$, and the root of $A_q$ to $v_{i_q}$ for all $q \in [p]$. Hence, the embedding $\phi$ restricted to the vertices of $B$ is the desired embedding of $B$ in $T$. 
\end{proof}

\begin{obs}\label{obs:fly}
Note that the knowledge that a vertex belongs to $F$ in the previous lemma is only needed when we reach it. So the set $F$ does not need being decided at the beginning of the procedure but can be decided on the fly.
\end{obs}

\subsubsection{Proof of Lemma~\ref{lem:stub}}
%%%%%%%%%%%%%%%%%%%%%%%

Let $A$ be a stub with $n$ nodes and $k$ leaves. 

Let  $B$ be the forest obtained from $A$ by removing the arcs and the internal vertices of the maximal directed paths of length at least $3$ contained in its segments.
The components of $B$ are called the {\bf islands} of $A$. Note that each island of $A$ contains at least one branch-node of $A$. Note moreover that there are at most $k-2$ islands in $A$. Let $\widehat B$ be the digraph whose vertices are the islands of $A$, and such that there is an arc from $C$ to $C'$ in $\widehat B$ if and only if in $A$ there is a directed out-path with origin in $C$ and terminus in $C'$. For all arc $e$ of $\widehat B$, we denote that directed out-path by $P(e)$. Observe that $B$ is a forest and $\widehat B$ is a tree. 

Choose an island $C_1$ that has indegree $0$. Take $C_1$ as the root of $\widehat B$. 
There is an ordering $(C_1,\dots,C_r)$ of the islands of $A$ such that
\begin{itemize} 
\item[(i)] if $C_p \rightarrow C_q$ then $p \le q$;
\item[(ii)] for each island $C$, there exist $p_C$ and $q_C$ such that an island $C_p$ is a descendant of $C$ in $\widehat B$ if and only if it verifies $p_C\leq p\leq q_C$.
\end{itemize}

For all $p \in [r]$, let $E^-(C_p)$ be the set of the downward arcs of $\widehat B$ with head $C_p$, and let $E^+(C_p)$ be the set of the upward arcs of $\widehat B$ with tail $C_p$.
%We will consider the paths $P(e)$ from the father to the son. 
For an arc $e\in E^+(C_p)$, we let $Q(e)$ be the path obtained from $P(e)$ by removing its last two vertices.
Similarly, for an arc $e\in E^-(C_p)$, we let $Q(e)$ be the path obtained from $\overline{P}(e)$ by removing its last two vertices.

 For all $p \in [r]$, the {\bf space} of $C_p$ is 
 $$\spc(C_p) = 12|C_p|+36k - 124 + \sum_{e \in E^-(C_p) \cup E^+(C_p)} (|Q(e)|+1).$$ 
 
 By definition we have  
 $$\sum_{p \in [r]} \left( |C_p| + \sum_{e  \in E^-(C_p) \cup E^+(C_p)} |Q(e)| \right) = n .$$
  Now
$ |E^-(C_p)|+|E^+(C_p)|$ is the number of arcs between $C_p$ and its sons in  $\widehat B$, so 
$$\sum_{p \in [r]} (|E^-(C_p)|+|E^+(C_p)|) \le r-1 \leq k-3.$$
Since $A$ is a stub, all its outer-segments have length $1$ and so remain in $B$. Moreover, an inner segment of $A$ is either directed, or  has two blocks with one of length $1$, or has three blocks with the first and the last of length $1$. Therefore, at most three of its internal vertices remain in $B$.

  Thus $\sum_{p\in[r]} |C_p| =|B| \leq k + (k-2) + 3(k-3) = 5k -11$.
  Consequently, 
  \begin{eqnarray*}
  \sum_{p\in [r]} \spc(C_p)  & \leq &  n + 11(5k-11)   + (k-3) +  (k-2)(36k -124)\\
 & \leq & \vstub .
  \end{eqnarray*}
  
  \medskip

 Let $T$ be a tournament of order $m=\vstub$, and let $(v_1, \dots , v_m)$ be a local median order of $T$.
 Now for all $i =1$ to $r$, reserve the first $\spc(C_p)$ unreserved vertices of $T$ for $C_p$. Therefore the set of vertices reserved for $C_p$ is
 $$R_p= \left\{~v_i~~\mid~~\sum_{q <p}\spc(C_q)+1 \leq i \leq \sum_{q \leq p}\spc(C_q) \right \}.$$ 

 Set $\displaystyle \alpha_p= \sum_{e \in E^-(C_p)}(|Q(e)|+1) + 6|C_p|+ 10k-29$. We partition $R_p$ into three sets.
 The {\bf middle} of $C_p$ is the set  $$M_p= \{~v_i~~\mid~~\sum_{q <p}\spc(C_q) + \alpha_p +1 \leq i \leq \sum_{q <p}\spc(C_q) + \alpha_p+16k-58\} ,$$
 the {\bf left margin} of $C_p$ is the set 
 $$M^-_p= \{~v_i~~\mid~~\sum_{q <p}\spc(C_q) +1 \leq i \leq \sum_{q <p}\spc(C_q) + \alpha_p\} ,$$
and the {\bf right margin} of $C_p$ is the set 
$$M^+_p= \{~v_i~~\mid~~\sum_{q <p}\spc(C_q) + \alpha_p+16k-57 \leq i \leq \sum_{q \leq p}\spc(C_q).$$

\medskip

We are going to build an embedding $\phi$ of $A$ into $T$.  
Run a Breadth-First Search algorithm on $B$, and let $\Pi$ be the resulting ordering.
The ordering $\Pi$ corresponds to a permutation $\pi$ of $[r]$: $\Pi=(C_{\pi(1)}, C_{\pi(2)}, \dots , C_{\pi(r)})$. 
The idea is to embed the islands in increasing order according of $\pi$ so that each island is treated before its sons.
When  a island $C_p$ is considered, we embed all the vertices of $A_p=C_p \cup \bigcup_{e \in E^-(C_p) \cup E^+(C_p)}Q(e)$ in $R_p$.
In the mean time, for each $e=C_pC_q$ in $E^+(C_p)$ (resp.~$E^-(C_p)$), we embed  the path between the terminus of $Q(e)$, which is the penultimate (resp.~second) node of $P(e)$, and the terminus (resp.~origin) of $P(e)$ in $M_q$ (this vertex is the root of $C_q$) using
Lemma~\ref{lem:lapin2}. When using this lemma, the internal vertex of this path is embedded in some vertex that must be forbidden for the others. Therefore, we need to keep track of this forbidden vertices in a set $F$.

Let us define formally the root $a_p$ of $C_p$.
Pick any node $a_1$ of $C_1$ as its root. 
 For all $p \in \{2,\dots,r\}$, let $C_q$ be the father of $C_p$ in $\widehat B$. There is an arc $e$ between $C_p$ and $C_q$.
 The root $a_p$ is the end of $P(e)$ which is in $C_p$.

A vertex is {\bf free} if it is not yet the image of a node.

Let us now describe the algorithm in detail. It keeps track of a set $F$ of at most $k-3$ vertices (at most one for each arc between two islands in $\widehat B$).  
To start, we set $F = \emptyset$, and we embed $a_1$ at $v_{\alpha_1+1}$.

Then for $t=1$ to $r$ do the following:
\begin{itemize}
\item[0.] Set $p=\pi^{-1}(t)$. The root $a_p$ of $C_p$ is already embedded at some vertex $v_i$.

\item[1.] Set $I_p= \{v_j \mid i-2|C_p|-2k+5 \leq j \leq i+2|C_p|+2k -5\}$. Embed $C_p$ in $T\langle I_p \rangle$ thanks to Lemma~\ref{lem:islands}, avoiding the vertices that are in $F \cap I_p$.

\item[2.] For each $e \in E^+(C_p)$, consider $P(e) = (x_{e,1},\dots ,x_{e,\ell_e})$ from $C_p$ to one of its sons $C_q$. Note that $x_{e,1}=a_p$ is already embedded, and that $x_{e,\ell_e}=a_q$. Consider the lowest integer $j \ge i+2|C_p|+2k-4$ such that $\phi(x_{e,1}) \rightarrow v_{j}$ and $v_{j}$ is free. Embed $x_{e,2}$ at $v_{j}$.

Proceed symmetrically, for the arcs in $E^-(C_p)$.

\item[3.] Apply Lemma~\ref{lem:arbos} on the paths $^*Q(e)$ for $e \in E^+(C_p)$. As noted previously, in Lemma~\ref{lem:arbos} we only need to know the vertices of $F$ when we reach them in the construction (since we can just reapply the algorithm on a slightly different arborescence when we meet a vertex in $F$). We consider the construction of these paths in order, and when the last node of $Q(e)$ is reached, do the following: apply Lemma~\ref{lem:lapin2} to get $4k-15$ internally disjoint  directed $2$-out-paths from $v_p$ to disjoint vertices in $M_{q}$ (with $C_q$ the head of $e$ in $\widehat B$); pick one such path that does not use any vertex of $F$ (here $|F| \le k-4$), nor any of the images of the roots of the $C_{q'}$ for $q' \in [r]\setminus \{p,q\}$ which are already embedded (there are at most $k-4$ of these); put its second vertex in $F$, embed the penultimate node of $P(e)$ at its second vertex, and embed the root of $C_q$ (which is also the terminus of $P(e)$) at its terminus. 
The vertex that was added to $F$ has a larger index than the vertices we are in the application of Lemma~\ref{lem:arbos}, so we can make sure to not embed another node at it.

Do the symmetric on the paths $^*Q(e)$ for $e \in E^-(C_p)$.
\end{itemize}

\bigskip

Let us know prove that this algorithm results in an embedding of $A$ into $T$.

\medskip
Let us first prove that every vertex is mapped to a vertex and that every vertex of $A_p$ is mapped into $R_p$.

At Step 1, we only embed the nodes of $C_p$ in $I_p$, which is in an interval of $4|C_p| + 4k-9$ vertices centered at some index $i$ in the middle $M_p$.

At Step 2, we hit at most  $|E^+(C_p)|$ out-neighbours of vertices that belong to $I_p$. Let $w_h=v_{i+2|C_p|+2k-5-h}$ be a vertex of $I_p$ (hence $0 \leq h \leq  4|C_p| + 4k-10$). 
It has at most $h$ out-neigbours in $I(w_h)=\{ v_j \mid i+2|C_p|+2k-4-h \leq j \leq i+2|C_p|+2k-5\}$. 
Set 
$$J^+_p= \{v_j \mid i+2|C_p|+2k-4 \leq j \leq  i +6 |C_p|+2|E^+(C_p)| +8k-21\}.$$ 
Note that $|J^+_p|= 2|E^+(C_p)|+ 4|C_p| + 6k -16$.
 By (M2), $w_h$ has at least $\frac{1}{2} (h + |J^+_p|) = \frac{1}{2}(h + 2|E^+(C_p)|+ 4|C_p| + 6k -16)\geq h+ |E^+(C_p)|+k-3$ out-neighbours in $I(w_h) \cup J^+_p$ and so at least $|E^+(C_p)|+k-3$ out-neighbours in $J^+_p$. 
Hence the vertices of $I_p$ have enough out-neighbours in $J^+_p$ to choose the $|E^+(C_p)|$ out-neighbours among vertices that are not in $F$.

Similarly, we hit at most $|E^-(C_p)|$ vertices that belong to the set $$J^-_p = \{v_j \mid i -6 |C_p|-2|E^-(C_p)| -8k+21 \leq j \leq  i-2|C_p|-2k+4\}.$$

At Step 3, we only need to ensure that the conditions of Lemma~\ref{lem:arbos} can be verified whithin $R_p$. Hence we only need to check that in $R_p$ there are $\sum_{e \in E^+(C_p)}(|Q(e)| - 1) +2|F|$ vertices after the last vertex of $J^+_p$ and  $\sum_{e \in E^-(C_p)}(|Q(e)| - 1) +2|F|$ vertices before the first vertex of $J^-_p$.
The vertex $v_i$ is in $M_p$, so $i\geq \sum_{q <p}\spc(C_q)  + \alpha_p+1$.
Hence, there are at least $\alpha_p - (6 |C_p|+2|E^-(C_p)| +8k-21) = \sum_{e \in E^-(C_p)}(|Q(e)| - 1) +2k-8 \geq  \sum_{e \in E^-(C_p)}(|Q(e)| - 1) +2|F|$ vertices before $v_i$. (Recall that $|F|\leq k-4$ when doing this).
Furthermore, $i\leq \sum_{q <p}\spc(C_q) + \alpha_p +16k -58$. So, in $R_p$, there are at least $\spc(C_p) - (\alpha_p + 16k-58+6 |C_p|+2|E^+(C_p)| +8k-29) = \sum_{e \in E^+(C_p)}(|Q(e)| - 1) +2k -8$ vertices after $v_i$.
This is what we wanted.

\medskip

Finally, let us now show that two nodes are never mapped to a same vertex.
%Suppose for a contradiction that two nodes $x$ and $y$ of $A$ are mapped to a same vertex.
%Let $t_x$ (resp.~$t_y$) be the value of $t$ at which $x$ (resp.~$y$) has been mapped.
%Without loss of generality, we may assume that $t_x \leq t_y$.

%Observe first that $t_x\neq t_y$. 
Observe first that at each loop, we map nodes on distinct vertices.
At Step 1, we map nodes into different vertices of $I_p$.
Then at Step 2, we embed nodes into distinct vertices of $J^+_p$ and $J^-_p$, and the three sets $I_p$, $J^+_p$ and $J^-_p$, are pairwise disjoint.
Finally, at Step 3, using Lemma~\ref{lem:arbos} and \ref{lem:lapin2}, we finish embedding the $P_e$ for $e\in E^+(C_p)$ into $M^+_p$ and embedding the $P_e$ for $e\in E^+(C_p)$ into $M^-_p$.  We take care of adding the second vertex of the $2$-paths in $F$ each time we apply  Lemma~\ref{lem:lapin2}, and that we always avoid embedding vertices in $F$. Hence two nodes embedded during a same loop are mapped to different vertices.

In addition, all vertices hit at the loop $t$ are in $R_p$ (with $p=\pi^{-1}(t)$) except the ones when applying Lemma~\ref{lem:lapin2}.
Since the $R_p$'s are pairwise disjoint, we only need to check that when applying Lemma~\ref{lem:lapin2}, we do not map a node to a vertex onto which another vertex was or will be mapped.  Property (ii) of the ordering $(C_1,\dots,C_r)$ implies that, when applying Lemma~\ref{lem:lapin2},  nodes are all mapped onto vertices in some $R_q$ such that $C_q$ is a descendant of $C_p$. So the only possible conflicts are with other vertices hit when applying this lemma. But we take care of avoiding those vertices (the second vertices of the paths of length $2$ are added to $F$ and we specifically avoid the root of $R_q$ for all $q \ne p$).

This completes the proof of Lemma~\ref{lem:stub}.

\section{Conclusion and further research}
%%%%%%%%%%%%%%%%%%%%%%

\subsection{Towards Conjecture~\ref{conj:n+k-1} and beyond}
%%%%%%%%%%%%%%%%%%%%%%%%%%%%%%%

The bound $\frac{3}{2}n + \frac{3}{2}k  -2$ of Theorem~\ref{t_tree2} can be replaced by $n+k-1 + \min_{r\in V(A)} \min(\gamma_r^{\uparrow},\gamma_r^{\downarrow})$. However, for any {\bf antidirected tree} $A$, that is an oriented tree in which every node has either in-degree $0$ or out-degree $0$, we have $\min_{r\in V(A)} \min(\gamma_r^{\uparrow},\gamma_r^{\downarrow})= \frac{1}{2}n + \frac{1}{2}k  -1$.

Another step towards Conjecture~\ref{conj:n+k-1} would be to prove them for antidirected trees.

\medskip

When proving Theorem~\ref{thm:veryfew}, we try to keep the proof as simple as possible and made not attempt to get the smallest upper bound on $g(k)$.
For example, we can improve on the bound $\vstub$ of Lemma~\ref{lem:stub} by studying more carefully on the size of $|F|$ at each loop.
Likewise, we can slightly improve Lemma~\ref{lem:reduc}. 
Doing so, we can get somewhat better upper bound on $g(k)$ than $\val$. However, all such bounds are quadratic in $k$, i.e. $\Omega (k^2)$.
A next step towards Conjecture~\ref{conj:n+k-1} would then be to prove that $g(k)\leq o(k^2)$ (that is every oriented tree of order $n$ with $k$ leaves is $(n+ o(k^2))$-unavoidable), and ideally that $g(k) \leq \alpha \cdot k$ for some absolute constant $\alpha$.

\medskip

Conjecture~\ref{conj:n+k-1} is tight because of the out-stars and in-stars. But those trees have few nodes: just one more than leaves.
In the same way, we believe that all the trees with $n$ nodes and $k$ leaves that are not $(n+k-2)$-unavoidable have $n$ small compared to $k$.

\begin{conjecture}\label{conj:n+k-2}
For every fixed integer $k$, there is an integer $n_k$ such that
 every oriented tree of order $n\geq n_k$ with $k$ leaves is $(n+k-2)$-unavoidable.
 \end{conjecture}

This conjecture holds for $k=2$ by a result of Havet and Thomass\'e~\cite{HaTh00}, and for $k=3$ as shown by Ceroi and Havet~\cite{Ceha04}.

\subsection{Generalisation to $k$-chromatic digraphs}
%%%%%%%%%%%%%%%%%%%%%%%%%%%

 A {\bf proper $k$-colouring} of a digraph is a mapping $c$ from its vertex into $\{1, \dots , k\}$ such that
$c(u)\neq c(v)$ for every arc $uv$. A digraph is {\bf $k$-colourable} if it admits a proper $k$-colouring.
The {\bf chromatic number} of a digraph $D$, denoted $\chi(D)$, is the least integer $k$ such that
$D$ is $k$-colourable. A digraph is {\it $k$-chromatic} if its chromatic number equals $k$.

The complete graph on $n$-vertices is the simplest $n$-chromatic graph, and so tournaments on $n$ vertices are the simplest $k$-chromatic digraphs.
The notion of unavoidability generalizes to the one of universality.
A digraph $F$ is {\bf $k$-universal} if it is contained in every digraph with chromatic number $k$.

Burr~\cite{Bur80} generalizes Sumner's conjecture to universality.
\begin{conjecture}[Burr~\cite{Bur80}]\label{conj:burr} 
Every every oriented tree of order $n$ is $(2n-2)$-universal.
\end{conjecture}

We also conjecture that Conjecture~\ref{conj:n+k-1} extends to universality.
\begin{conjecture}\label{conj:n+k-1-univ}
 Every oriented tree of order $n$ with $k$ leaves is $(n+k-1)$-universal.
 \end{conjecture}

%\bibliographystyle{abbrv}
%\bibliography{unavoid}

\end{document}